\documentclass[twocolumn,letterpaper, 10 pt,journal]{IEEEtran}

\usepackage{amsmath,amssymb,amsfonts,amsthm}
\usepackage{graphicx}
\usepackage[mathscr]{euscript}
\usepackage[dvipsnames]{xcolor}
\usepackage[colorlinks=true,allcolors=MidnightBlue]{hyperref}
\usepackage[switch]{lineno}
\usepackage{enumitem}
\usepackage[bb=boondox]{mathalfa}

\newcommand{\bmat}[1]{\begin{bmatrix}#1\end{bmatrix}}
\newcommand{\smat}[1]{\left[\begin{smallmatrix}#1\end{smallmatrix}\right]}

\newcommand{\As}{{A^\star}}
\newcommand{\Bs}{{B^\star}}
\newcommand{\tu}[1]{\textup{#1}}
\newcommand{\real}{\mathbb{R}}

\newtheorem{assumption}{Assumption}
\newtheorem{problem}{Problem}
\newtheorem{theorem}{Theorem}
\newtheorem{proposition}{Proposition}
\newtheorem{remark}{Remark}
\newtheorem{lemma}{Lemma}

\allowdisplaybreaks
\graphicspath{{./figures-state/}}
\begin{document}

\title{Controller synthesis for input-state data\\ with measurement errors}
\author{Andrea Bisoffi, Lidong Li, Claudio De Persis, Nima Monshizadeh
\thanks{A. Bisoffi ({\tt\footnotesize andrea.bisoffi@polimi.it}) is with Department of Electronics, Information, and Bioengineering, Politecnico di Milano, 20133, Italy.
C. De Persis, L. Li, N. Monshizadeh are with Engineering and Technology Institute, University of Groningen, 9747AG, The Netherlands.
This work was supported in part by the Project Digital Twin of the Research Programme TTW Perspective which is (partly) financed by the Dutch Research Council (NWO) under Grant P18-03. The work of L. Li was supported by the Chinese Scholarship Council.}
}

\maketitle

\begin{abstract}
We consider the problem of designing a state-feedback controller for a linear system, based only on noisy input-state data.
We focus on input-state data corrupted by measurement errors, which, albeit less investigated, are as relevant as process disturbances in applications.
For energy and instantaneous bounds on these measurement errors, we derive linear matrix inequalities for controller design where the one for the energy bound is equivalent to robust stabilization of all systems consistent with the noisy data points via a common Lyapunov function.
\end{abstract}


\section{Introduction}
\label{sec:intro}

We would like to design a state feedback controller for a discrete-time linear time-invariant (LTI) system without knowing the parameter matrices $(A,B)$ of its state equation, but with only input-state measurements.
When such measurements are noise-free and enjoy persistence of excitation, using them for this goal is not so interesting since these measurements identify $(A, B)$ exactly and the data-based design boils down to a model-based one.
We thus focus on noisy measurements.
In practical settings, it is reasonable and desirable 
that the noise in the measurements is bounded in some sense.
The data generation mechanism and the noise bound yield a set of parameter matrices consistent with data, as in set membership identification \cite{fogel1979system}.
Since the actual system is indistinguishable from all others in the set, our goal is to asymptotically stabilize all systems consistent with data.
We share this approach with many recent works on (direct) data driven control  \cite{depersis2020formulas,berberich2021combining,coulson2021distributionally,vanwaarde2022Slemma,celi2023closed}.

Within this framework, however, the noise can enter the system in different ways when data are generated and different bounds on the noise can be postulated.
As for the first aspect, a large part of the recent literature on data-driven control has considered a process disturbance, which captures unmodeled dynamics in the state equation \cite{berberich2021combining,vanwaarde2022Slemma,bisoffi2022petersen}.
Still, less attention has been devoted to disturbances corrupting the input, which capture actuator inaccuracies, or corrupting the state or the output, which capture sensor inaccuracies.
The close link with actuator\slash sensor inaccuracies makes it relevant to consider disturbances corrupting input, state or output, which are known as errors-in-variables in system identification \cite{soderstrom2003errors,cerone2011set,soderstrom2018errors}. In a static setting, they correspond to distinguishing between manifest and latent variables. We refer to \cite[\S 1.1]{soderstrom2018errors} for further motivation of errors-in-variables, which we term measurement errors here.
Measurement errors are commonly postulated to have an energy bound or an instantaneous bound \cite[\S I]{bertsekas1971recursive}:
the former considers the whole sequence of the errors acting during data collection, see \eqref{energy_bound} below, whereas the latter considers each of such error instances, see \eqref{inst_bound} below.
The treatment of instantaneous bounds is relevant for analogous advantages to those evidenced in \cite{bisoffi2021tradeoffs} for process disturbances.

Based on this discussion, the case of measurement errors with instantaneous bounds is particularly relevant for control applications, since such bounds on input or state errors can be inferred based on actuator or sensor characteristics, pointwise in time, whereas it is less trivial to infer tight energy bounds on process disturbances.
We provide linear matrix inequalities (LMIs) to design a state-feedback controller for the setting of measurement errors with energy and instantaneous bounds.

\emph{Related literature.}
In \cite{bisoffi2021tradeoffs} and \cite{bisoffi2022petersen}, data is generated by an LTI system affected by a process disturbance. Here, we depart from that setting and consider instead input-state data affected by measurement errors, in an error-in-variables setting.
In \cite{depersis2020formulas}, sufficient conditions for controller design are given when the measurement error on the state satisfies an energy bound, whereas we provide necessary and sufficient conditions here.
Measurement errors and a process disturbance, which satisfy an instantaneous bound (in $\infty$-norm), are considered in \cite{miller2022errorinvariables}. To handle bilinearity in the set of system parameters, \cite{miller2022errorinvariables} formulates the controller design as a polynomial optimization problem, which is approximated by a converging sequence of sum-of-squares programs.
Here, Proposition~\ref{prop:matrix_elim_lemma} enables controller design by solving a single LMI.
In \cite{abuabiah2023noniterative}, a controller design for multi-input-multi-output linear systems from input-output data is proposed; this design imposes matching a reference model via a nonconvex program \cite[\S III.D]{abuabiah2023noniterative} and stability of the so-designed controller is checked a posteriori \cite[\S V]{abuabiah2023noniterative}.
Finally, the results in this paper complement those in~\cite{otherpaperarxiv}.
While \cite{otherpaperarxiv} addresses the case of output feedback, with input and output data
corrupted by measurement errors, we consider state feedback here.
Focusing on the special case of state feedback, we give stronger conditions than those in \cite{otherpaperarxiv}, namely a necessary and sufficient condition in Theorem~\ref{thm:design} instead of the sufficient conditions in \cite[Thm.~1]{otherpaperarxiv}.
Besides, \cite{otherpaperarxiv} treats only energy bounds whereas we treat also instantaneous ones here.

\emph{Contribution.}
For measurement errors with energy and instantaneous bounds, we obtain two corresponding LMIs in Theorems~\ref{thm:design} and \ref{thm:design_inst}.
These LMIs depend only on the collected data and the postulated bounds, and enable the design of a state-feedback controller.
Such controller is guaranteed to asymptotically stabilize all systems consistent with data, and the actual one in particular.
Importantly, Theorem~\ref{thm:design} takes the form of a necessary and sufficient condition.
We also provide the independently relevant Proposition~\ref{prop:matrix_elim_lemma}, which can be interpreted as a matrix elimination result.

\section{Preliminaries}

\subsection{Notation}\label{sec:notation}

The identity matrix of size $n$ and the zero matrix of size $m \times n$ are $I_n$ and $0_{m \times n}$: the indices are dropped when no confusion arises.
The largest eigenvalue of a symmetric matrix $M$ is $\lambda_{\max}(M)$.
The largest and smallest singular values of a matrix $M$ are $\sigma_{\max}(M)$ and $\sigma_{\min}(M)$.
The 2-norm of a vector $v$ is $|v|$.
The induced 2-norm of a matrix $M$ is $\|M\|$ and is equivalent to $\sigma_{\max}(M)$.
For matrices $M$, $N$ and $O$ of compatible dimensions, we abbreviate $MNO(MN)^{\top}$ to $MN \cdot O [\star]^{\top}$, where the dot in the second expression clarifies unambiguously that $MN$ is the term to be transposed.
For a symmetric matrix $\left[\begin{smallmatrix} M & N \\ N^{\top} & O \end{smallmatrix}\right]$, we may use the shorthand writing $\left[\begin{smallmatrix} M & N \\ \star & O \end{smallmatrix}\right]$.
For a positive semidefinite matrix $M$, $M^{1/2}$ is its unique positive semidefinite square root \cite[p.~440]{horn2013matrix}.

\subsection{Matrix elimination result}

The next result is key for the sequel.

\begin{proposition}
\label{prop:matrix_elim_lemma}
Consider matrices $E \in \real^{n_1 \times n_2}$, $F\in\real^{n_1 \times n_3}$, $G \in \real^{n_3 \times n_3}$ with $G = G^\top \succeq 0$. Then,
\begin{subequations}
\begin{align}
\label{matrix_elim_lemma_ineq}
E E^\top \preceq F G F^\top
\end{align}
if and only if there exists $D \in \real^{n_3 \times n_2}$ such that
\begin{align}
\label{matrix_elim_lemma_eq_ineq}
E = F D,\quad D D^\top \preceq G.
\end{align}
\end{subequations}
\end{proposition}
\begin{proof}
($\Longleftarrow$) From~\eqref{matrix_elim_lemma_eq_ineq}, $E E^\top = F D D^\top F^\top \preceq F G F^\top$. \newline
($\Longrightarrow$) We need to show that if \eqref{matrix_elim_lemma_ineq} holds, then there exists $D$ such that \eqref{matrix_elim_lemma_eq_ineq} holds.
If $G=0$, it must be $E=0$ and $D=0$ yields \eqref{matrix_elim_lemma_eq_ineq}.
We then consider $G \neq 0$ in the rest of the proof.
For $G \neq 0$, there exist a real orthogonal matrix $U$ (i.e., $U U^{\top} = U^{\top} U = I$) and a diagonal positive definite $\Lambda_1$ such that an eigendecomposition of $G$ is
\begin{align}
\label{eigendecomp_G}
\!\!G = U \smat{\Lambda_1 & 0\\ 0 & 0} U^\top = \bmat{U_1~U_2} \smat{\Lambda_1 & 0\\ 0 & 0} \smat{U_1^\top\\ U_2^\top}= U_1 \Lambda_1 U_1^\top, 
\end{align}
which becomes $G = U \Lambda_1 U^\top$ for $G \succ 0$.
By the eigendecomposition of $G$, \eqref{matrix_elim_lemma_ineq} is equivalent to
\begin{equation}
\label{EEtop<=FU1Lambda1.top}
    E E^{\top} \preceq F U_1 \Lambda_1 U_1^{\top} F^{\top}.
\end{equation}
If $F U_1 = 0$, it must be from~\eqref{EEtop<=FU1Lambda1.top} that $E = 0$; any $D  = U_1 D_1$ with $D_1 D_1^\top \preceq \Lambda_1$ satisfies $E = 0 = F U_1 D_1 = F D$ and $D D^\top = U_1 D_1 D_1^\top U_1^\top \preceq U_1 \Lambda_1 U_1^\top = G$, which amounts to  \eqref{matrix_elim_lemma_eq_ineq}.
Otherwise, if $F U_1 \neq 0$, we can define the next quantities:
\begin{itemize}[noitemsep,nosep,left=0pt]
\item a nonsingular matrix $V$ such that
\begin{align}
V F U_1 = \smat{ \hat F_1 \\ 0} \label{VFU1=}
\end{align}
with $\hat F_1$ full row rank; $V$ can be obtained as the reduced row echelon form \cite[p.~11-12]{horn2013matrix};
\item $\smat{\hat E_1 \\ \hat E_2} := VE$ where $\hat E_1$ has as many rows as $\hat F_1$;
\item $ \hat F_1^{\tu{R}} := \Lambda_1 \hat F_1^\top (\hat F_1 \Lambda_1 \hat F_1^\top)^{-1}$ where $\hat F_1 \Lambda_1 \hat F_1^\top \succ 0$ (it is thus invertible) since $\hat F_1$ has full row rank and $\Lambda_1 \succ 0$;
\item $D_1 := \hat F_1^{\tu{R}} \hat E_1$.
\end{itemize}
We now show that $D =U_1 D_1$ satisfies \eqref{matrix_elim_lemma_eq_ineq}.
Since $V$ is nonsingular, \eqref{EEtop<=FU1Lambda1.top} holds if and only if
\begin{align}
& V E E^\top V^\top \preceq V F U_1 \Lambda_1 U_1^\top F^\top V^\top \notag \\
& \overset{\eqref{VFU1=}}{\iff} 
\smat{\hat E_1\\ \hat E_2} \smat{\hat E_1\\ \hat E_2}^\top \preceq \smat{\hat F_1\\ 0} \Lambda_1 \smat{\hat F_1\\ 0}^\top \notag \\
& \iff \smat{\hat E_1 \hat E_1^\top & \hat E_1 \hat E_2^\top\\ \hat E_2 \hat E_1^\top & \hat E_2 \hat E_2^\top} \preceq \smat{\hat F_1 \Lambda_1 \hat F_1^\top & 0 \\ 0 & 0} \notag \\
& \iff \hat E_1 \hat E_1^\top \preceq \hat F_1 \Lambda_1 \hat F_1^\top \text{ and } \hat E_2 = 0. \label{E1E1top}
\end{align}
Note that $\hat F_1^{\tu{R}}$ is a right inverse of $\hat F_1$ because $\hat F_1 \cdot \hat F_1^{\tu{R}} = \hat F_1 \cdot \Lambda_1 \hat F_1^\top (\hat F_1 \Lambda_1 \hat F_1^\top)^{-1} = I$.
Hence, $D_1$ satisfies
\begin{align}
& \hat E_1 = \hat F_1 \hat F_1^{\tu{R}} \hat E_1 = \hat F_1 D_1 \label{hatE1=}\\
& D_1 D_1^\top = \hat F_1^{\tu{R}} \hat E_1 \hat E_1^\top (\hat F_1^{\tu{R}})^\top \overset{\eqref{E1E1top}}{\preceq} \hat F_1^{\tu{R}} \hat F_1 \Lambda_1 \hat F_1^\top (\hat F_1^{\tu{R}})^\top \notag \\
& = \Lambda_1^{1/2} \Big(\Lambda_1^{-1/2} \hat F_1^{\tu{R}} \hat F_1 \Lambda_1^{1/2}\Big)\Big(\Lambda_1^{1/2} \hat F_1^\top (\hat F_1^{\tu{R}})^\top \Lambda_1^{-1/2}\Big) \Lambda_1^{1/2} \notag \\
& =: \Lambda_1^{1/2} W W^\top \Lambda_1^{1/2}.\label{D1D1top}
\end{align}
The so-defined $W$ is symmetric (i.e., $W = W^\top$ by the definition of $\hat F_1^{\tu{R}}$), is a projection (i.e., $W^2=W$ \cite[p.~38]{horn2013matrix}), and thus each of its eigenvalues is 0 or 1 \cite[1.1.P5]{horn2013matrix} so that
\begin{align*}
\sigma_{\max}(W) & \! =\! \sqrt{\lambda_{\max}(W W^{\top})} \! =\! \sqrt{\lambda_{\max}(W W)} \! =\! \lambda_{\max}(W) \! =\! 1
\end{align*}
and then $W W^\top \preceq I$.
We can then conclude from~\eqref{D1D1top} that
\begin{align}
D_1 D_1^\top \preceq \Lambda_1^{1/2} W W^\top \Lambda_1^{1/2} \preceq \Lambda_1. \label{D1D1top_reform}
\end{align}
Therefore,
\begin{align}
& V F D = V F U_1 D_1 \overset{\eqref{VFU1=}}{=} \smat{\hat F_1\\ 0} D_1 = \smat{\hat  F_1 D_1\\ 0} \overset{\eqref{hatE1=}}{=} \smat{\hat E_1\\ 0}  \notag \\  
& \overset{\eqref{E1E1top}}{=} \smat{\hat E_1\\ \hat E_2} = V E \iff  F D = E , \label{FD=E}
\end{align}
by $V$ nonsingular, and
\begin{align}
D D^\top = U_1 D_1 D_1^\top U_1^\top \overset{\eqref{D1D1top_reform}}{\preceq} U_1 \Lambda_1 U_1^\top \overset{\eqref{eigendecomp_G}}{=} G. \label{DDtop}
\end{align}
\eqref{FD=E} and \eqref{DDtop} correspond to \eqref{matrix_elim_lemma_eq_ineq}.
\end{proof}

Because going from \eqref{matrix_elim_lemma_eq_ineq} to \eqref{matrix_elim_lemma_ineq} dispenses with matrix $D$, Proposition~\ref{prop:matrix_elim_lemma} can be interpreted as a matrix elimination result.

\section{Problem formulation}
\label{sec:probl_form}

Consider the discrete-time linear time-invariant system
\begin{align}
x^+ & = \As x + \Bs u \label{sys_x_star}
\end{align}
with state $x \in \real^n$ and input $u \in \real^m$.
The matrices $\As$ and $\Bs$ are unknown to us and, instead of their knowledge, we rely on collecting input-state data with measurement errors to design a controller for~\eqref{sys_x_star}, as we explain below.

Input-state data are collected by performing an experiment on~\eqref{sys_x_star}.
Consider
\begin{align}
\label{u_m,y_m}
u^{\tu{m}} := u + e_u \text{ and } x^{\tu{m}} := x + e_x
\end{align}
where the measured input $u^{\tu{m}}$ differs from the actual input $u$ of~\eqref{sys_x_star} by an unknown error $e_u$ and the measured state $x^{\tu{m}}$ differs from the actual state $x$ of~\eqref{sys_x_star} by an unknown error $e_x$.
The data collection experiment is depicted in Figure~\ref{fig:data_collection_exp_state} and is then as follows: for $k=0, \dots, T-1$, apply the signal $u^{\tu{m}}(k)$; along with error $e_u(k)$, this results in (unknown) input $u(k) = u^{\tu{m}}(k) - e_u(k)$ and (unknown) state $x(k)$ for some initial condition $x(0)$; measure the signal $x^{\tu{m}}(k) = x(k) + e_x(k)$.
The available data, on which our control design is based, are then 
$\{ u^{\tu{m}}(k)\}_{k=0}^{T-1}$, $\{x^{\tu{m}}(k) \}_{k=0}^T$.
Note that we consider a measurement error on state $x$, unlike~\cite{bisoffi2021tradeoffs,bisoffi2022petersen} where a process disturbance $d$ on the dynamics was considered and the data generation mechanism was $x^+ = A^\star x + B^\star u + d$.

\begin{figure}
\centerline{\includegraphics[scale=0.8]{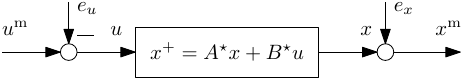}}
\caption{Scheme of data collection experiment for input-state measurements.}
\label{fig:data_collection_exp_state}
\end{figure}

The collected data satisfy, for $k=0, \dots, T-1$,
\begin{subequations}
\label{exp_data_vec}
\begin{align}
& x^{\tu{m}}(k+1) = \As
x^{\tu{m}}(k)
+ \Bs
u^{\tu{m}}(k) \notag \\
& \hspace*{19mm} + e_x(k+1) - \As
e_x(k)
- \Bs 
e_u(k) \\
&  = 
[\As \ \Bs] \smat{x^{\tu{m}}(k)\\ u^{\tu{m}}(k)} + [I \ - \! \As \ - \! \Bs]\epsilon(k)
\end{align}
with 
\begin{align*}
\epsilon(k) := \smat{e_x(k+1)\\ e_x(k)\\ e_u(k)}.
\end{align*}
\end{subequations}
For the sequel, the $T$ equalities in~\eqref{exp_data_vec} are equivalent to
\begin{align}
X_1^{\tu{m}}
& = [\As \ \Bs] \smat{X_0^{\tu{m}}\\ U_0^{\tu{m}}} + [I \ - \! \As \ - \! \Bs] \, E_{10} \label{exp_data_mat}
\end{align}
with definitions
\begin{subequations}
\begin{align}
X_1^{\tu{m}} & := \bmat{x^{\tu{m}}(1) & x^{\tu{m}}(2) & \dots & x^{\tu{m}}(T)}\\
X_0^{\tu{m}} & := \bmat{x^{\tu{m}}(0) & x^{\tu{m}}(1) & \dots & x^{\tu{m}}(T-1)} \\
U_0^{\tu{m}} & := \bmat{u^{\tu{m}}(0) & u^{\tu{m}}(1) & \dots & u^{\tu{m}}(T-1)} \\
E_{10}  & := \bmat{\epsilon(0)\hspace*{8.5pt} & \epsilon(1)\hspace*{7.2pt} & \dots & \epsilon(T-1)}.  \label{E10}
\end{align}
\end{subequations}
The quantities $X_1^{\tu{m}}$, $X_0^{\tu{m}}$ and $U_0^{\tu{m}}$ are obtained from available data $\{ u^{\tu{m}}(k)\}_{k=0}^{T-1}$, $\{x^{\tu{m}}(k) \}_{k=0}^T$
whereas $E_{10}$ is unknown.

As prior information, in addition to~\eqref{exp_data_mat}, we consider, in parallel, two types of bounds that the error $\epsilon$ in~\eqref{exp_data_vec} can satisfy.
The first bound is an energy bound: this assumes that the ``energy'' of the whole sequence of vectors $\epsilon(0)$, \dots, $\epsilon(T-1)$ is bounded by some matrix $\Theta = \Theta^\top \succeq 0$ as
\begin{align}
\label{energy_bound}
E_{10} E_{10}^\top
& =
\sum_{k=0}^{T-1}
\epsilon(k)
\epsilon(k)^\top
\preceq 
\Theta.
\end{align}
In other words, the prior information is that any error sequence acting during data collection, in particular the actual $E_{10} = [\epsilon(0) \ \epsilon(1) \ \dots \ \epsilon(T-1)]$,
belongs to the set
\begin{equation}
\label{setE}
\mathcal{E}_{\tu{e}} :=  \{ E  \in  \real^{ (2 n + m) \times T} : E E^\top \preceq \Theta \}.
\end{equation}
The second, alternative, bound is an instantaneous bound: this assumes that each vector $\epsilon(0)$, \dots, $\epsilon(T-1)$ is bounded in norm by some scalar $\theta \ge 0$ as
\begin{align}
\label{inst_bound}
| \epsilon(k) |^2 \le \theta, \quad k = 0, \dots, T-1.
\end{align}
In other words, the prior information is that any error vector, in particular the actual $\epsilon(0)$, \dots, $\epsilon(T-1)$, belongs to the set
\begin{equation}
\label{setE_inst}
\mathcal{E}_{\tu{i}} :=  \{ \varepsilon  \in  \real^{ 2 n + m} : |\varepsilon|^2 \le \theta \}.
\end{equation}
We note that imposing $\Theta = 0$ or $\theta=0$ yields the setting of noise-free data as an immediate special case.

Although the actual parameters $[\As \ \Bs]$ from~\eqref{sys_x_star} are unknown, each of the previous two bounds allows us to characterize the set of parameters $[A \ B]$ consistent with \eqref{setE} and data in~\eqref{exp_data_mat} or with \eqref{setE_inst} and data in~\eqref{exp_data_vec}.
For the energy bound in~\eqref{setE}, the set is
\begin{align}\label{setC}
& \mathcal{C}_{\tu{e}} := \{ [A \ B] \in \real^{n \times (n+m)} \colon   \\
& \quad X_1^{\tu{m}}  = [A \ B] \smat{X_0^{\tu{m}}\\ U_0^{\tu{m}}} + [I \ - \! A \ - \! B] \, E, \ E \in \mathcal{E}_{\tu{e}}  \},  \notag
\end{align}
cf.~\eqref{exp_data_mat}.
For the instantaneous bound in~\eqref{setE_inst}, the set is
\begin{subequations}\label{setC_instwithABC}
\begin{equation}
\label{setC_cap}
\mathcal{C}_{\tu{i}} := \bigcap_{k=0}^{T-1} \mathcal{C}_{\tu{i}}^k
\end{equation}
where the set $\mathcal{C}_{\tu{i}}^k$ of parameters consistent with the data point at $k=0$, \dots, $T-1$ is defined as
\begin{align}
& \mathcal{C}_{\tu{i}}^k := \{ [A \ B] \in \real^{n \times (n+m)} \colon \label{setC_inst} \\
& \,\,
x^{\tu{m}}(k+1) = 
[A \ B] \smat{x^{\tu{m}}(k)\\ u^{\tu{m}}(k)} + [I \ - \! A \ - \! B] \, \varepsilon, \ \varepsilon \in \mathcal{E}_{\tu{i}} \}, \notag
\end{align} 
\end{subequations}
cf.~\eqref{exp_data_vec}.
We emphasize that $[\As \ \Bs] \in \mathcal{C}_{\tu{e}}$ since $E_{10} \in \mathcal{E}_{\tu{e}}$; $[\As \ \Bs] \in \mathcal{C}_{\tu{i}}$ since $\epsilon(0) \in \mathcal{E}_{\tu{i}}$, \dots, $\epsilon(T-1) \in \mathcal{E}_{\tu{i}}$.
We make the next remark on the sets $\mathcal{C}_{\tu{e}}$ and $\mathcal{C}_{\tu{i}}$.
\begin{remark}
\label{rem:additional_structure}
In $\mathcal{C}_{\tu{e}}$ or $\mathcal{C}_{\tu{i}}$, we do not rely on the additional structure data in~\eqref{exp_data_vec} possess (for any $k$, the first component of $\epsilon(k)$ is equal to the second component of $\epsilon(k+1)$) because the same treatment could be given for data points from multiple experiments and, in that case, such additional structure may not be present.$\hfill\triangleleft$
\end{remark}

Our goal is to control~\eqref{sys_x_star} and make the origin asymptotically stable by using the feedback law
\begin{equation}
\label{u=Kx}
u = K x
\end{equation}
for some matrix $K$ to be designed.
In principle, if $[\As \ \Bs]$ were known in~\eqref{sys_x_star}, we would like to render $\As + \Bs K$ Schur.
In lieu of the knowledge of $[\As \ \Bs]$, we need to exploit the information available from data and embedded in the sets $\mathcal{C}_{\tu{e}}$ or $\mathcal{C}_{\tu{i}}$.
Thus, we set out to design $K$ such that  $A+B K$ is certified to be Schur for all $[A \ B] \in \mathcal{C}_{\tu{e}}$ or for all $[A \ B] \in \mathcal{C}_{\tu{i}}$ by a common Lyapunov function, as in quadratic stabilization \cite{barmish1985necessary}.
This is respectively equivalent to the robust control problems
\begin{subequations}\label{rob_contr_probl}
\begin{align}
& \text{find} \  K, P = P^\top \succ 0 \\
& \text{s.t.} \ \ (A+B K) P (A+B K)^\top \! - \! P \prec 0 \ \ \forall [A \ B] \in \mathcal{C}_{\tu{e}} \label{rob_contr_probl:ineq}
\end{align}
\end{subequations}
or
\begin{subequations}\label{rob_contr_probl_inst}
\begin{align}
& \text{find} \  K, P = P^\top \succ 0 \\
& \text{s.t.} \ \  (A+B K) P (A+B K)^\top \! - \! P \prec 0 \ \ \forall [A \ B] \in \mathcal{C}_{\tu{i}}. \label{rob_contr_probl_inst:ineq}
\end{align}
\end{subequations}

To summarize, next are our problem statements.
\begin{problem}
\label{probl}
With collected data $\{ u^{\tu{m}}(k)\}_{k=0}^{T-1}$, $\{x^{\tu{m}}(k) \}_{k=0}^T$
and with the error sequence satisfying the energy bound $\mathcal{E}_{\tu{e}}$, design a matrix $K$ to solve \eqref{rob_contr_probl} and so ensure that the origin of the feedback interconnection of~\eqref{sys_x_star} and \eqref{u=Kx} is globally asymptotically stable.
\end{problem}
\begin{problem}
\label{probl_inst}
With collected data $\{ u^{\tu{m}}(k)\}_{k=0}^{T-1}$, $\{x^{\tu{m}}(k) \}_{k=0}^T$
and with the error sequence satisfying the instantaneous bound $\mathcal{E}_{\tu{i}}$, design a matrix $K$ to solve \eqref{rob_contr_probl_inst} and so ensure that the origin of the feedback interconnection of~\eqref{sys_x_star} and \eqref{u=Kx} is globally asymptotically stable.
\end{problem}

In the next remark, we provide a typical yet conservative way to obtain an energy bound from an instantaneous one.
Still, instantaneous bounds on measurement errors capture actuator or sensor characteristics and are thus more relevant in practice, as discussed in Section~\ref{sec:intro}.

\begin{remark}\label{rem:conversion_to_energy_bound}
Suppose we know that for some $\bar{e}_x \ge 0$ and $\bar{e}_u \ge 0$, any errors $e_x$ and $e_u$ satisfy $|e_x|^2 \le \bar{e}_x$ and $|e_u|^2 \le \bar{e}_u$.
Then, for each $k =0,\dots, T-1$,
\begin{align*}
\left|
\smat{
e_x(k+1)\\
e_x(k)\\
e_u(k)
}
\right|^2 \! \! = \!
|e_x(k+1)|^2 \! + |e_x(k)|^2 \! + |e_u(k)|^2\! \le 2 \bar{e}_x +\bar{e}_u
\end{align*}
and, from~\eqref{E10},
\begin{align*}
E_{10} E_{10}^\top
& 
\preceq 
\sum_{k=0}^{T-1}
\left|
\smat{
e_x(k+1)\\
e_x(k)\\
e_u(k)
}\right|^2 I 
\preceq T (2 \bar{e}_x +\bar{e}_u) I.
\end{align*}
In this way, we can take $\Theta$ as $T (2 \bar{e}_x +\bar{e}_u) I$.
$\hfill\triangleleft$
\end{remark}

\section{Controller design for energy bound}
\label{sec:design}

In this section we solve Problem~\ref{probl}.
To this end, consider the set $\mathcal{C}_{\tu{e}}$ in~\eqref{setC}, i.e.,
\begin{align}
\mathcal{C}_{\tu{e}} & = \{ [A \ B] \colon \label{setC:almost_def}  \\
& X_1^{\tu{m}} - [A \ B] \smat{X_0^{\tu{m}}\\ U_0^{\tu{m}}} = [I \ - \! A \ - \! B] \, E , \ E E^\top \preceq \Theta \}. \notag
\end{align}
Importantly, Proposition~\ref{prop:matrix_elim_lemma} allows rewriting $\mathcal{C}_{\tu{e}}$ equivalently as
\begin{align*}
\mathcal{C}_{\tu{e}} & = \big\{ [A \ B] \colon \\
& (X_1^{\tu{m}} - [A \ B] \smat{X_0^{\tu{m}}\\ U_0^{\tu{m}}} )
\cdot [\star]^\top \preceq [I \ - \! A \ - \! B] \cdot \Theta [\star]^\top \big\} .
\end{align*}

Partition $\Theta = \Theta^\top \succeq 0$ as
\begin{equation}
\label{Theta_partitioned}
\Theta =: \bmat{\Theta_{11} & \Theta_{12}\\ \Theta_{12}^\top & \Theta_{22}}
\end{equation}
with $\Theta_{11} = \Theta_{11}^\top \in \real^{n \times n}$ and $\Theta_{22} = \Theta_{22}^\top \in \real^{(n+m) \times  (n+m)}$.
With~\eqref{Theta_partitioned}, algebraic computations rewrite the set $\mathcal{C}_{\tu{e}}$ as
\begin{subequations}\label{setCwithABC}
\begin{align}
\mathcal{C}_{\tu{e}}  =  \{ Z \! = \! [A \ B] & \colon  Z  \mathscr{A}  Z^\top \! + \! Z \mathscr{B}^\top \! + \! \mathscr{B} Z^\top \! + \! \mathscr{C} \preceq 0\} \label{setC:ABC} \\
\mathscr{A} & := \smat{X_0^{\tu{m}}\\ U_0^{\tu{m}}} \smat{X_0^{\tu{m}}\\ U_0^{\tu{m}}}^\top - \Theta_{22}, \\
\mathscr{B} & :=  -  X_1^{\tu{m}} \smat{X_0^{\tu{m}}\\ U_0^{\tu{m}}}^\top + \Theta_{12},\\
\mathscr{C} & := X_1^{\tu{m}} {X_1^{\tu{m}}}^\top  - \Theta_{11}.
\end{align}
\end{subequations}

To effectively work with $\mathcal{C}_{\tu{e}}$, we make the next assumption on the collected data.
\begin{assumption}
\label{ass:Psi0}
$\smat{X_0^{\tu{m}}\\ U_0^{\tu{m}}} \smat{X_0^{\tu{m}}\\ U_0^{\tu{m}}}^\top \succ \Theta_{22}$.
\end{assumption}

This assumption is of signal-to-noise-ratio type. Indeed, for
\begin{align*}
S_0 := \smat{X_0\\ U_0} := 
\smat{
x(0) & x(1) & \dots & x(T-1)\\
u(0) & u(1) & \dots & u(T-1)
},
\end{align*}%
by \cite[Lemma~10]{otherpaperarxiv}, if $\sigma_{\min}(S_0 S_0^\top)/\sigma_{\max}(\Theta_{22}) > 4$, Assumption~\ref{ass:Psi0} holds, see \cite{otherpaperarxiv} for more details.

Assumption~\ref{ass:Psi0} amounts to requiring $\mathscr{A} \succ 0$. Hence, $\mathscr{A}^{-1}$ exists and we can define
\begin{align}
& \mathscr{Z} := - \mathscr{B} \mathscr{A}^{-1}, \quad \mathscr{Q} := \mathscr{B} \mathscr{A}^{-1} \mathscr{B}^\top - \mathscr{C}. \label{setC:ZQ}
\end{align}
Thanks to Assumption~\ref{ass:Psi0} and $[\As \ \Bs] \in \mathcal{C}_{\tu{e}}$, the set $\mathcal{C}_{\tu{e}}$ has the properties summarized in the next result.
\begin{lemma}
\label{lemma:cons_asmpt_for_set_C}
Under Assumption~\ref{ass:Psi0}, we have that:
\begin{align}
&\hspace*{-6pt} \mathcal{C}_{\tu{e}} = \big\{ Z\colon (Z - \mathscr{Z}) \mathscr{A} (Z - \mathscr{Z})^\top \preceq \mathscr{Q} \big\}, \label{setC_ZAQ}\\
&\hspace*{-6pt} \mathscr{Q} \succeq 0,\\
&\hspace*{-6pt} \mathcal{C}_{\tu{e}} = \big\{ \mathscr{Z} + \mathscr{Q}^{1/2} \Upsilon \mathscr{A}^{-1/2} \colon \Upsilon \in \real^{n\times (n+m)}, \Upsilon \Upsilon^\top \preceq  I_n \big\}, \label{setC_unitBall}
\end{align}
and $\mathcal{C}_{\tu{e}}$ is bounded with respect to any matrix norm.
\end{lemma}
\begin{proof}
If $\mathscr{A} \succ 0$, $\mathscr{Z}$ and $\mathscr{Q}$ in~\eqref{setC:ZQ} are well-defined and algebraic computations yield \eqref{setC_ZAQ} from~\eqref{setC:ABC}.
By $[\As \ \Bs] \in \mathcal{C}_{\tu{e}}$, we have from~\eqref{setC_ZAQ} that $\mathscr{Q} \succeq ([\As \ \Bs] - \mathscr{Z}) \mathscr{A} ([\As \ \Bs] - \mathscr{Z})^\top \succeq 0$ by $\mathscr{A} \succ 0$.
$\mathscr{A} \succ 0$ and $\mathscr{Q} \succeq 0$ allow applying \cite[Proposition~1]{bisoffi2022petersen} to  obtain \eqref{setC_unitBall} from~\eqref{setC_ZAQ} and applying \cite[Lemma 2]{bisoffi2022petersen} to show with analogous arguments that the nonempty $\mathcal{C}_{\tu{e}}$ is bounded with respect to any matrix norm.
\end{proof}

With Lemma~\ref{lemma:cons_asmpt_for_set_C}, we can design a controller solving Problem~\ref{probl} next.

\begin{theorem}
\label{thm:design}
For data $\{ u^{\tu{m}}(k)\}_{k=0}^{T-1}$, $\{x^{\tu{m}}(k) \}_{k=0}^T$, suppose Assumption~\ref{ass:Psi0} holds.
Feasibility of%
\begin{subequations}
\label{rob_contr_probl_LMI}
\begin{align}
& \text{find} & & Y, P = P^\top \succ 0 \label{rob_contr_probl_LMI:find}\\
& \text{s. t.} & & 
\bmat{
-P -\mathscr{C} & 0 & \mathscr{B}\\
0 & - P & \smat{P & Y^\top} \\
\mathscr{B}^\top & \smat{P\\ Y} & -\mathscr{A}} \prec 0 \label{rob_contr_probl_LMI:ineq}
\end{align}
\end{subequations}
is equivalent to feasibility of~\eqref{rob_contr_probl}.
If \eqref{rob_contr_probl_LMI} is feasible, a controller gain $K$ satisfying \eqref{rob_contr_probl} is $K = Y P^{-1}$.
Moreover, $x= 0$ is globally asymptotically stable for the feedback interconnection of unknown plant $x^+ = A^\star x + B^\star u$ and controller $u = K x$. 
\end{theorem}
\begin{proof}
By $P \succ 0$ and Schur complement, \eqref{rob_contr_probl:ineq} is equivalently
\begin{align}
& \bmat{
- P &  -[A \ B] \smat{I\\ K} P \\
\star & -P
} \prec 0 \quad \forall [A \ B] \in \mathcal{C}_{\tu{e}}. 
\end{align}
There exist $K$ and $P = P^\top \succ 0$ satisfying this condition if and only if there exist $Y$ and $P = P^\top \succ 0$ satisfying
\begin{align*}
& \bmat{
- P & - Z \smat{P\\ Y}  \\
\star & -P
} \prec 0 \quad \forall Z \in \mathcal{C}_{\tu{e}},
\end{align*}
with $Y = K P$.
By $\mathcal{C}_{\tu{e}}$ in~\eqref{setC_unitBall}, obtained under Assumption~\ref{ass:Psi0}, this condition holds if and only if
\begin{align*}
& 0 \succ 
\bmat{
- P & -\mathscr{Z} \smat{P\\ Y} \\
\star & -P
}+
\bmat{ \mathscr{Q}^{1/2}\\
0}
\Upsilon
\bmat{
0 & -\mathscr{A}^{-1/2} \smat{P\\ Y}} \\
& 
+
\bmat{
0\\
-\smat{P\\ Y}^\top \mathscr{A}^{-1/2}}
\Upsilon^\top
\bmat{ \mathscr{Q}^{1/2} & 0} \quad \forall \Upsilon \text{ with } \| \Upsilon \| \le 1.
\end{align*}
By Petersen's lemma reported in~\cite[Fact~1]{bisoffi2022petersen}, this condition holds if and only if there exists $\lambda > 0$ such that
\begin{align*}
& 0 \succ 
\bmat{
- P & -\mathscr{Z} \smat{P\\ Y} \\
\star & -P
} +
\frac{1}{\lambda}
\bmat{ \mathscr{Q}^{1/2}\\
0}
\bmat{ \mathscr{Q}^{1/2}\\
0}^\top \\
& +
\lambda
\bmat{
0\\
-\smat{P\\ Y}^\top \mathscr{A}^{-1/2}}
\bmat{
0\\
-\smat{P\\ Y}^\top \mathscr{A}^{-1/2}}^\top \\
& =
\bmat{
- P + \frac{1}{\lambda} \mathscr{Q} & -\mathscr{Z} \smat{P\\ Y} \\
\star & -P + \lambda \smat{P\\ Y}^\top \mathscr{A}^{-1} \smat{P\\ Y}
}.
\end{align*}
The existence of $Y$, $P = P^\top \succ 0$, $\lambda >0$ such that this inequality holds is equivalent to the existence of $Y$, $P = P^\top \succ 0$ such that
\begin{align}
& 0 \succ 
\bmat{
- P + \mathscr{Q} &  -\mathscr{Z} \smat{P\\ Y} \\
\star & -P + \smat{P\\ Y}^\top \mathscr{A}^{-1} \smat{P\\ Y} 
}. \label{LMIinZQA}
\end{align}
By the definitions of $\mathscr{Q}$ and $\mathscr{Z}$ in~\eqref{setC:ZQ}, this inequality is equivalent to
\begin{align*}
0 & \succ 
\bmat{
- P - \mathscr{C} & 0 \\
\star & -P
}
+
\bmat{
\mathscr{B}\\
\smat{P\\ Y}^\top
}
\mathscr{A}^{-1}
\bmat{
\mathscr{B}\\
\smat{P\\ Y}^\top
}^\top
\end{align*}
and, by Schur complement, to~\eqref{rob_contr_probl_LMI:ineq}.
Since \eqref{rob_contr_probl_LMI} ensures that $A + B K$ is Schur for all $[A \ B] \in \mathcal{C}_{\tu{e}}$ and $[\As \ \Bs] \in \mathcal{C}_{\tu{e}}$, $x= 0$ is globally asymptotically stable for $x^+ = (A^\star + B^\star K) x$.
\end{proof}

Theorem~\ref{thm:design} shows that, since the set $\mathcal{C}_{\tu{e}}$ of parameters consistent with data can be characterized without conservatism, robust stabilization of all matrices in the set $\mathcal{C}_{\tu{e}}$ (with a common Lyapunov function) is equivalent to the linear matrix inequality in \eqref{rob_contr_probl_LMI}.
We also note that \eqref{rob_contr_probl_LMI:ineq} is equivalent, by~\eqref{LMIinZQA}, Assumption~\ref{ass:Psi0} and Schur complement, to
\begin{align*}
& 0 \succ 
\bmat{
- P + \mathscr{Q} & - \mathscr{Z} \smat{P\\ Y} & 0\\
\star & -P & \smat{P & Y^\top}\\
\star & \star & -\mathscr{A}
}.
\end{align*}

\section{Controller design for instantaneous bound}
\label{sec:design_inst}

In this section we solve Problem~\ref{probl_inst}.
To this end, consider the set $\mathcal{C}_{\tu{i}}^k$ in~\eqref{setC_inst} for $k=0$, \dots, $T-1$, i.e.,
\begin{align*}
\mathcal{C}_{\tu{i}}^k & = \big\{ [A \ B] \colon x^{\tu{m}}(k+1) - [A \ B] \smat{x^{\tu{m}}(k)\\ u^{\tu{m}}(k)}  \\
& \hspace*{21mm}
= 
[I \ - \! A \ - \! B] \, \varepsilon, \ \varepsilon \varepsilon^\top \preceq \theta I \big\}.
\end{align*}
Importantly, Proposition~\ref{prop:matrix_elim_lemma} allows rewriting $\mathcal{C}_{\tu{i}}^k$ at $k=0$, \dots, $T-1$ equivalently as
\begin{align*}
& \mathcal{C}_{\tu{i}}^k =  \big\{ [A \ B] \colon \big(x^{\tu{m}}(k+1) - [A \ B] \smat{x^{\tu{m}}(k)\\ u^{\tu{m}}(k)}\big) \cdot [\star]^\top   \\
& \hspace*{10mm} \preceq [I \ - \! A \ - \! B] \cdot (\theta I) [\star]^\top = [I \ A \ B] \cdot (\theta I) [\star]^\top \big\}.
\end{align*}%
Hence, $[A \ B] \in \mathcal{C}_{\tu{i}}$ in~\eqref{setC_cap} if and only if
\begin{align}
&  \text{ for } k = 0, \dots, T-1, \notag\\
&  \smat{I\\ A^\top\\ B^\top}^\top \!\!
\bigg\{ \!\!
\smat{x^{\tu{m}}(k+1)\\ -x^{\tu{m}}(k)\\ -u^{\tu{m}}(k)} \!\!
\cdot [\star]^\top \!\!\!
- \! \smat{
\theta I_n & 0 & 0 \\
0 & \theta I_n & 0 \\
0 & 0 & \theta I_m }
\!\! \bigg\} \!
\smat{I\\ A^\top\\ B^\top} \preceq 0. \label{in_Ci_equiv}
\end{align}
We can design a controller solving Problem~\ref{probl_inst} next.
\begin{theorem}
\label{thm:design_inst}
For data $\{ u^{\tu{m}}(k)\}_{k=0}^{T-1}$, $\{x^{\tu{m}}(k) \}_{k=0}^T$, feasibility~of%
\begin{subequations}%
\label{rob_contr_probl_LMI_inst}
\begin{align}
& \text{find} & & Y, P = P^\top \succ 0, \tau_0 \ge 0, \dots, \tau_{T-1} \ge 0 \label{rob_contr_probl_LMI_inst:find}\\
& \text{s. t.} & & 
\!\! \rule{0pt}{24pt} 0 \succ \smat{
-P & 0 & 0 & 0 \\
0  & P & Y^\top & 0 \\
0 & Y & 0 & Y\\
0 & 0& Y^\top & -P} \label{rob_contr_probl_LMI_inst:ineq}  \\
& & & \!\! - \sum_{k=0}^{T-1} \tau_k \Bigg( 
\smat{x^{\tu{m}}(k+1)\\ -x^{\tu{m}}(k)\\ -u^{\tu{m}}(k)\\ 0}
\smat{x^{\tu{m}}(k+1)\\ -x^{\tu{m}}(k)\\ -u^{\tu{m}}(k)\\ 0}^\top\!\!
- \smat{\theta I_n & 0 & 0 & 0\\
0 & \theta I_n & 0 & 0\\
0 & 0 & \theta I_m & 0\\
0 & 0 & 0 & 0}
\Bigg) \notag 
\end{align}
\end{subequations}
implies feasibility of \eqref{rob_contr_probl_inst}.
If \eqref{rob_contr_probl_LMI_inst} is feasible, a controller gain satisfying \eqref{rob_contr_probl_inst} is $K = Y P^{-1}$.
Moreover, $x= 0$ is globally asymptotically stable for the feedback interconnection of unknown plant $x^+ = A^\star x + B^\star u$ and controller $u = K x$.
\end{theorem}
\begin{proof}
By Schur complement, \eqref{rob_contr_probl_LMI_inst:ineq} is equivalent to
\begin{align*}
& 0 \succ 
\smat{
-P & 0 & 0\\
0  & P & Y^\top\\
0 & Y & Y P^{-1} Y^\top} \\
& - \sum_{k=0}^{T-1} \tau_k \Bigg( 
\smat{x^{\tu{m}}(k+1)\\ -x^{\tu{m}}(k)\\ -u^{\tu{m}}(k)}
\smat{x^{\tu{m}}(k+1)\\ -x^{\tu{m}}(k)\\ -u^{\tu{m}}(k)}^\top
- \smat{
\theta I_n & 0 & 0 \\
0 & \theta I_n & 0 \\
0 & 0 & \theta I_m }
\Bigg).
\end{align*}
By the change of variables $Y = KP$ between $Y$ and $K$, the previous inequality is equivalent to
\begingroup
\thinmuskip=.8mu plus 1mu minus 1mu
\medmuskip=1mu plus 1mu minus 1mu
\thickmuskip=1.2mu plus 1mu minus 1mu
\begin{align}
& 0 \succ M := \label{M_def} \\
& \smat{
-P & 0 & 0\\
0  & P & P K^\top\\
0 & KP & K P K^\top} - \sum_{k=0}^{T-1} \tau_k \bigg( 
\smat{x^{\tu{m}}(k+1)\\ -x^{\tu{m}}(k)\\ -u^{\tu{m}}(k)}
\cdot [\star]^{\top}
- \smat{
\theta I_n & 0 & 0 \\
0 & \theta I_n & 0 \\
0 & 0 & \theta I_m }
\bigg). \notag
\end{align}
\endgroup
Since $[I\ A \ B]$ has full row rank for each $[A \ B]$, the previous inequality implies \cite[Obs.~7.1.8]{horn2013matrix} that
\begin{align*}
& 0 \succ
\bmat{I & A & B }
M
\bmat{I & A & B }^\top \quad \forall [A \ B].
\end{align*}
By $\tau_0 \ge 0$, \dots, $\tau_{T-1} \ge 0$, the previous condition implies that
\begin{align*}
& 0 \succ
\smat{I\\ A^\top\\ B^\top}^\top \!
\smat{
-P & 0 & 0\\
0  & P & P K^\top\\
0 & KP & K P K^\top}\!
\smat{I\\ A^\top\\ B^\top} \, \forall [A \ B] \text{ s. t. \eqref{in_Ci_equiv} holds.}
\end{align*}
This is equivalent to: $0 \succ (A+BK) P (A+BK)^\top -P$ for all $[A \ B] \in \mathcal{C}_{\tu{i}}$, which is \eqref{rob_contr_probl_inst:ineq}.
\end{proof}

Theorem~\ref{thm:design_inst} shows that a sufficient condition for robust stabilization of all matrices in the set $\mathcal{C}_{\tu{i}}$ (with a common Lyapunov function) is feasibility of the linear matrix inequality in~\eqref{rob_contr_probl_LMI_inst}.
Moreover, if \eqref{rob_contr_probl_LMI_inst:ineq} is feasible, we have that for some $\tau_0 \ge 0$, \dots, $\tau_{T-1} \ge 0$,
\begin{align}
\label{nec_cond_inst}
\sum_{k=0}^{T-1} \tau_k \Big(\smat{x^{\tu{m}}(k)\\ u^{\tu{m}}(k)}\smat{x^{\tu{m}}(k)\\ u^{\tu{m}}(k)}^\top-\theta I_{n+m}\Big) \succ 0.
\end{align}
This is true since \eqref{rob_contr_probl_LMI_inst:ineq} is equivalent to~\eqref{M_def}, and \eqref{M_def} implies negative definiteness of the principal submatrix of $M$ with elements $(2,2)$, $(2,3)$, $(3,2)$, $(3,3)$.
\eqref{nec_cond_inst} can be interpreted as a signal-to-noise ratio.

\section{Numerical results}
\label{sec:example}

Data are generated by a system corresponding to the discretization of a simple distillation column borrowed from~\cite{albertos2006multivariable}.
Specifically, the \emph{unknown} matrices $A^\star$ and $B^\star$ are selected equal to \cite[$A$ and $B$ on p.~95]{albertos2006multivariable}.
Such system has $n=7$, $m=3$ and the eigenvalues of the unknown $A^\star$ are $0$, $0$, $0.8607$, $0.8607$, $0.9024$, $0.9024$, $0.9217$.

We assume to know that each errors $e_x$ and $e_u$ satisfy $|e_x|^2 \le \bar{e}_x$ and $|e_u|^2 \le \bar{e}_u$ for some $\bar{e}_x \ge 0$ and $\bar{e}_u \ge 0$.
We consider an experiment with $\bar{e}_x=5\cdot 10^{-5}$, $\bar{e}_u=5\cdot 10^{-5}$, $T=200$.
By converting the instantaneous bounds into an energy bound as in Remark~\ref{rem:conversion_to_energy_bound}, we obtain $\Theta = T(2\bar{e}_x + \bar{e}_u)I_{17}$;
we also obtain $\theta = 2\bar{e}_x + \bar{e}_u$, see Remark~\ref{rem:conversion_to_energy_bound} and \eqref{setE_inst}.
For this $\Theta$ and $\theta$, \eqref{rob_contr_probl_LMI} was not feasible whereas \eqref{rob_contr_probl_LMI_inst} was.
The controller designed with \eqref{rob_contr_probl_LMI_inst} is
\begin{align*}
K \!=\!
\smat{
   -0.2759&    0.1518&    1.2911&   -0.7883&   -0.0065&   -2.8121&    0.8276\\
    0.1886&   -1.1926&  -19.8129&    2.7679&    0.4695&   30.2622&   -9.6385\\
    0.4625&   -0.0028&   -1.3624&    0.4232&   -2.4954&    1.0034&   -0.4562
}\!
\end{align*}
and the resulting closed-loop solutions are in Figure~\ref{fig:cl}.

\begin{figure}
\centerline{\includegraphics[scale=0.6]{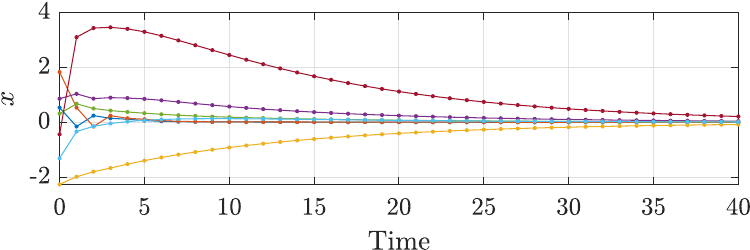}}
\caption{Closed loop of~\eqref{sys_x_star} and $u = K x$, where $K$ is designed by \eqref{rob_contr_probl_LMI_inst}.}
\label{fig:cl}
\end{figure}

We now evaluate the performance of the two previous approaches quantitatively, given the same $\bar{e}_x$ and $\bar{e}_u$.
To do so, we follow the type of analysis in~\cite{bisoffi2021tradeoffs}, to which we refer the reader for more details.
Consider different numbers of data points $T \in \{20, 40, 60, 80, 100, 120, 140, 160, 180, 200\}$ and different bounds $\theta \in \{
10^{-6},\sqrt{10}\cdot10^{-6},10^{-5},\sqrt{10}\cdot10^{-5},$ $10^{-4},\sqrt{10}\cdot10^{-4},10^{-3} \}$. 
Based on these values of $\theta$, we set $\bar{e}_x = \bar{e}_u = \theta/3$ so that for each $e_x$, $|e_x| \le \sqrt{\bar{e}_x} \in \{0.0006,0.0010,0.0018, 0.0032,0.0058,0.0103,0.0183 \}$, and $\Theta = T \theta I$.
For each of these pairs $(T,\theta)$, we randomly generate 20 data sequences; for each data sequence, apply \eqref{rob_contr_probl_LMI} in Theorem~\ref{thm:design} and \eqref{rob_contr_probl_LMI_inst} in Theorem~\ref{thm:design_inst}; then count the number $n_{\tu{feas}}$ of instances when \eqref{rob_contr_probl_LMI} and \eqref{rob_contr_probl_LMI_inst} are feasible.
For both cases, we plot the ratio $n_{\tu{feas}}/20 \in [0,1]$ in Figure~\ref{fig:feasPerc} as a function of $(T,\theta)$, with a logarithmic scale for $\theta$.
In line with~\cite{bisoffi2021tradeoffs}, Figure~\ref{fig:feasPerc} shows that it is more convenient to employ the instantaneous bound directly together with the sufficient condition in Theorem~\ref{thm:design_inst}, rather than convert it into an energy bound to use the necessary and sufficient condition in Theorem~\ref{thm:design}.
The downside of Theorem~\ref{thm:design_inst} over Theorem~\ref{thm:design} is that the former involves $T$ more decision variables, but this is hardly an issue unless $T$ is quite large.

\begin{figure}
\centerline{\includegraphics[scale=0.65]{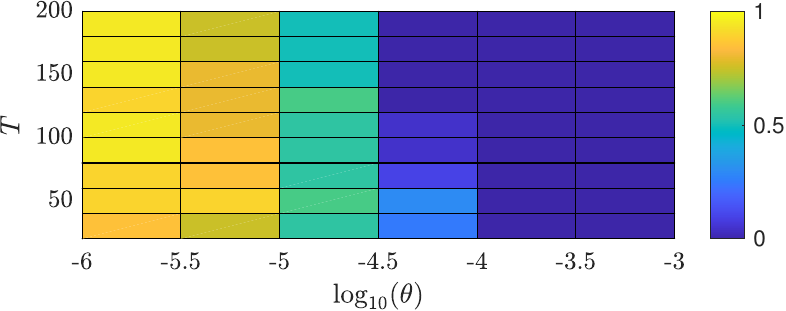}}
\centerline{\includegraphics[scale=0.65]{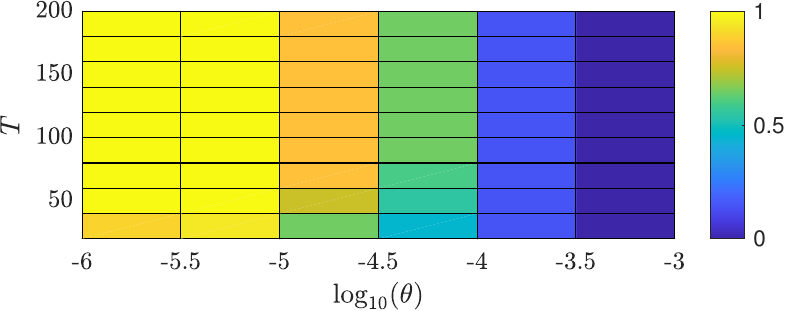}}
\caption{Ratio of $n_{\tu{feas}}/20 \in [0,1]$ as a function of $(T,\log_{10}(\theta))$ for \eqref{rob_contr_probl_LMI} (top) and \eqref{rob_contr_probl_LMI_inst} (bottom), with ``good'' regions in yellow.}
\label{fig:feasPerc}
\end{figure}

\section{Conclusions}
\label{sec:concl}
We have addressed the problem of designing a state-feedback controller based only on noisy data.
Specific to this work is that we have considered the setting where input and state  are affected by measurements errors with both energy and instantaneous bounds.
We have provided two linear matrix inequalities for the design of a controller that asymptotically stabilizes all systems consistent with the data points and the respective energy or instantaneous bounds.
For the energy bound, the linear matrix inequality is actually equivalent to robust stabilization with a common Lyapunov function.
Numerical examples have validated these results and provided a caveat for when energy bounds are derived by converting instantaneous ones.
Interesting directions for future work are extensions to nonlinear systems and control problems beyond stabilization.

\bibliographystyle{IEEEtran}
\bibliography{pubs-state}

\end{document}